\documentclass[conference]{IEEEtran}
\IEEEoverridecommandlockouts
\usepackage{cite}
\usepackage{amsmath,amssymb,amsfonts,amsthm}
\usepackage{algorithmic}
\usepackage{graphicx}
\usepackage{textcomp}
\usepackage{xcolor}
\allowdisplaybreaks
\usepackage{balance}
\def\BibTeX{{\rm B\kern-.05em{\sc i\kern-.025em b}\kern-.08em
    T\kern-.1667em\lower.7ex\hbox{E}\kern-.125emX}}
    \newtheorem{lemma}{Lemma}
    
    \newtheorem*{remark}{Remark}

\begin{document}

\title{On the R\'{e}nyi Cross-Entropy$^*$
\thanks{\textsuperscript{*}This work was supported in part by NSERC of Canada.}
}

\author{\IEEEauthorblockN{Ferenc Cole Thierrin, Fady Alajaji, and Tam\'{a}s Linder} \smallskip
\IEEEauthorblockA{\small Department of Mathematics and Statistics \\ Queen's University \\Kingston, ON K7L 3N6, Canada \\ 
Emails: \{14fngt, fa, tamas.linder\}@queensu.ca}
}
\maketitle

\begin{abstract}
The R\'{e}nyi cross-entropy measure between two distributions, a generalization of the Shannon cross-entropy, was recently used as a loss function for the improved design of deep learning generative adversarial networks. In this work, we examine the properties of this measure and derive closed-form expressions for it when one of the distributions is fixed and when both distributions belong to the exponential family. We also analytically determine a formula for the cross-entropy rate for stationary Gaussian processes and for finite-alphabet Markov sources.
\end{abstract}


\section{Introduction}
The R\'{e}nyi entropy \cite{renyi} of order $\alpha$ of a discrete distribution (probability mass function) $p$ with finite support $\mathbb{S}$, defined as $$H_\alpha(p)=\frac{1}{1-\alpha}\ln \sum_{x\in\mathbb{S}}p(x)^\alpha $$ for $\alpha>0,\alpha\neq 1$, is a generalization of the Shannon entropy,\footnote{For ease of reference, a table summarising the Shannon entropy and cross-entropy measures as well as the Kullback-Liebler (KL) divergence is provided in Appendix~A.} $H(p)$, in that $\lim_{\alpha\to1}H_\alpha(p)=H(p)$. Similarly, the R\'{e}nyi divergence (of order $\alpha$) between two discrete distributions $p$ and $q$ with common finite support $\mathbb{S}$, given by $$D_\alpha(p||q)= \frac{1}{\alpha-1}\ln \sum_{x\in\mathbb{S}} p(x)^\alpha q(x)^{1-\alpha},$$ reduces to the KL divergence, $D(p\|q)$, as $\alpha \to 1$. 

Since the introduction of these measures, several other R\'{e}nyi-type information measures have been put forward, each obeying the condition that their limit as $\alpha$ goes to one reduces to a Shannon-type information measure (e.g., see \cite{mutualInfo} and the references therein for three different order $\alpha$ extensions of Shannon's mutual information due to Sibson, Arimoto and Csisz\'{a}r.)

Many of these definitions admit natural counterparts in the (absolutely) continuous case (i.e., when the involved distributions have a probability density function (pdf)), giving rise to information measures such as the R\'{e}nyi differential entropy for pdf $p$ with support $\mathbb{S}$, $$h_\alpha(p)=\frac{1}{1-\alpha}\ln \int_\mathbb{S}p(x)^\alpha \, dx,$$ and the R\'{e}nyi (differential) divergence between pdfs $p$ and $q$ with common support $\mathbb{S}$, $$D_\alpha(p||q)= \frac{1}{\alpha-1}\ln \int_\mathbb{S} p(x)^\alpha q(x)^{1-\alpha} \, dx.$$

The R\'{e}nyi cross-entropy between distributions $p$ and $q$ is an analogous generalization of the Shannon cross-entropy $H(p;q)$. Two definitions for this measure have been suggested. In \cite{rgan}, mirroring the fact that Shannon's cross-entropy satisfies $H(p;q)=D(p\|q)+H(p)$, the authors define R\'{e}nyi cross-entropy as 
\begin{align}\label{rgan-ce} \tilde H_\alpha(p;q):=D_\alpha(p||q)+H_\alpha (p).
\end{align}
In contrast, prior to \cite{rgan}, the authors of \cite{Alba} introduced the R\'{e}nyi cross-entropy in their study of the  so-called shifted R\'{e}nyi measures (expressed as the logarithm of weighted generalized power means). Specifically, upon simplifying Definition~6 in~\cite{Alba}, their expression for the R\'{e}nyi cross-entropy between distributions $p$ and $q$ is  given by 
\begin{equation}\label{alba-ce}
H_\alpha\left(p;q\right):=\frac{1}{1-\alpha}\ln\sum_{x\in\mathbb{S}} p\left(x\right)q\left(x\right)^{\alpha-1}.    
\end{equation}
For the continuous case, the definition in \eqref{alba-ce} can be readily converted to yield the R\'{e}nyi differential cross-entropy between pdfs $p$ and $q$:
\begin{equation} \label{renyi-diff-ce}
h_\alpha(p;q):=\frac{1}{1-\alpha}\ln \int_\mathbb{S} p(x)q(x)^{\alpha-1} \, dx.
\end{equation}

As the R\'{e}nyi differential divergence and entropy were already calculated for numerous distributions in \cite{Gil} and \cite{Song}, respectively, determining the R\'{e}nyi differential cross-entropy using the definition in \eqref{rgan-ce} is straightforward. As such, this paper's focus is to establish closed-form expressions of the R\'{e}nyi differential cross-entropy as defined in \eqref{renyi-diff-ce} for various distributions, as well as to derive the R\'{e}nyi cross-entropy rate for two important classes of sources with memory, Gaussian and Markov sources.

Motivation for determining formulae for the R\'{e}nyi cross-entropy extends beyond idle curiosity. The Shannon differential cross-entropy was used as a loss function for the design of deep learning generative adversarial networks (GANs) \cite{Goodfellow2014}. Recently, the R\'{e}nyi differential cross-entropy measures in \eqref{renyi-diff-ce} and \eqref{rgan-ce}, were used in \cite{paper,Bhatia} and \cite{rgan}, respectively, to generalize the original GAN loss function.
It is shown that in \cite{paper} and \cite{Bhatia} that the resulting R\'{e}nyi-centric generalized loss function
preserves the equilibrium point satisfied by the original GAN based on the Jensen-R\'{e}nyi divergence \cite{kluza19}, a natural extension of the Jensen-Shannon divergence \cite{jsd}.
In \cite{rgan}, a different R\'{e}nyi-type generalized loss function is obtained and is shown to benefit from stability properties. Improved stability and system performance are shown in \cite{paper,Bhatia} and \cite{rgan} by virtue of the $\alpha$ parameter that can be judiciously used to fine-tune the adopted generalized loss functions which recover the original GAN loss function as $\alpha\to1$. 

The rest of this paper is organised as follows. In Section II, basic properties of the R\'{e}nyi cross-entropy are examined. In Section~III, the R\'{e}nyi differential cross-entropy for members of the exponential family is calculated. In Section IV, the R\'{e}nyi differential cross-entropy between two different distributions is obtained. In Section V, the R\'{e}nyi differential cross-entropy rate is derived for stationary Gaussian sources. Finally in Section~VI, the R\'{e}nyi cross-entropy rate is established for finite-alphabet time-invariant Markov sources.

\medskip

\section{Basic Properties of the R\'{e}nyi cross-entropy and differential cross-entropy}

For the R\'{e}nyi cross-entropy $H_\alpha(p;q)$ to deserve its name it would be preferable that it satisfies at least two key properties: it reduces to the R\'{e}nyi entropy when $p=q$ and its limit as $\alpha$ goes to one is the Shannon cross-entropy. Similarly, it is desirable that the R\'{e}nyi differential cross-entropy $h_\alpha(p;q)$ reduces to the R\'{e}nyi differential entropy when $p=q$ and its limit as $\alpha$ tends to one yields the Shannon differential cross-entropy. In both cases, the former property is trivial, and the latter property was proven in \cite{Bhatia} for the continuous case under some finiteness conditions (in the discrete case, the result holds directly via L'H\^{o}pital's rule).

It is also proven in \cite{Bhatia} that the R\'{e}nyi differential cross-entropy $h_\alpha(p;q)$ is non-increasing in $\alpha$ by showing that its derivative with respect to  $\alpha$ is non-positive. The same monotonicity property holds in the disrcrete case. 

Like its Shannon counterpart, the R\'{e}nyi cross-entropy is non-negative ($H_\alpha(p;q)\ge 0$); while the R\'{e}nyi differential cross-entropy can be negative.  This is easily verified when, for example, $\alpha=2$ and $p$ and $q$ are both Gaussian (normal) distributions with zero mean and variance $1/(8\sqrt{\pi})$, and parallels the same lack of non-negativity of the Shannon differential cross-entropy. 

We close this section by deriving the cross-entropy limit, $\lim_{\alpha \to \infty}H_\alpha(p;q)$.
To begin with, for any non-zero constant $\tilde{c}$, we have
\begin{align}
&\lim_{\alpha\to\infty}\frac{1}{1-\alpha}\ln \sum_{x\in\mathbb{S}} \tilde{c} q\left(x\right)^{\alpha-1} \nonumber\\
&=\lim_{\alpha\to\infty}\frac{1}{1-\alpha}\ln \tilde{c} + \lim_{\alpha\to\infty}\frac{1}{1-\alpha}\ln\sum_{x\in\mathbb{S}} q\left(x\right)^{\alpha-1} \nonumber\\
&=\lim_{\beta\to\infty}\frac{1-\beta}{-\beta}\frac{1}{1-\beta}\ln\sum_\mathbb{S} q\left(x\right)^{\beta} \qquad (\beta=\alpha-1) \nonumber\\
&= \lim_{\beta\to\infty}H_\beta (q)
= -\ln q_M, \label{max-limit} 
\end{align}
where $q_M:=\max_{x\in \mathbb{S}}q(x)$ and where we have used the fact that for the R\'{e}nyi entropy, $\lim_{\alpha\to\infty}H_\alpha(q)=-\ln q_M$.
Now, denoting the minimum and maximum values of $p(x)$ over $\mathbb{S}$ by $p_m$ and $p_M$, respectively, we have that for $\alpha>1$,
\begin{align*}
    \frac{1}{1-\alpha}\ln\sum_{x\in\mathbb{S}} p_m q\left(x\right)^{\alpha-1}&\leq\frac{1}{1-\alpha}\ln\sum_{x\in\mathbb{S}} p(x) q\left(x\right)^{\alpha-1}\\ 
    \text{and }\\
    \frac{1}{1-\alpha}\ln\sum_{x\in\mathbb{S}} p(x) q\left(x\right)^{\alpha-1}&\leq\frac{1}{1-\alpha}\ln\sum_{x\in\mathbb{S}} p_M q\left(x\right)^{\alpha-1}, 
\end{align*}
and hence by \eqref{max-limit} we obtain
\begin{equation}
\lim_{\alpha\to\infty} H_\alpha\left(p;q\right)= -\ln q_M.
\end{equation}

\medskip

\section{R\'{e}nyi Differential Cross-Entropy for Exponential Family Distributions}

A probability distribution on $\mathbb{R}$ or $\mathbb{R}^n$ with parameter $\theta$ is said to belong to the exponential family (e.g., see \cite{casella2002statistical}) if on its support $\mathbb{S}$ it admits a pdf of the form 
\begin{equation} \label{exp-f1}
f(x)=c (\theta)b (x)\exp \left(\eta(\theta)\cdot T (x)\right), \qquad x \in \mathbb{S},
\end{equation} 
for some real-valued (measurable) functions $c$, $b$, $\eta$ and $T$.\footnote{Note that $\theta$ and consequently $T(x)$ can be vectors in cases where the distribution admits multiple parameters.} Here $\eta$ is known as the natural parameter of the distribution, $T(x)$ is the sufficient statistic and $c(\theta)$ is the normalization constant in the sense that for all $\theta$ within the parameter space $$\int_\mathbb{S} b (x)\exp \left(\eta(\theta)\cdot T (x)\right) dx= c(\theta)^{-1}.$$
The pdf in \eqref{exp-f1} can also be written as 
\begin{equation} \label{exp-f2}
    \displaystyle{f(x) =b (x)\exp \left(\eta\cdot T (x)+A(\eta)\right)}, 
    \end{equation}
 where $A\left(\eta(\theta)\right)=\ln c(\theta) $. Examples of distributions in the exponential family include the Gaussian, Beta, and exponential distributions.

\begin{lemma}\label{lemma1}
Let $f_1(x)$ and $f_2(x)$ be pdfs of the same type in the exponential family with natural parameters $\eta_1$ and $\eta_2$, respectively. Define $f_h(x)$ as being of the same type as $f_1$ and $f_2$ but with natural parameter $\eta_h= \eta_1 + (\alpha-1)\eta_2$.  Then 
\begin{equation}
h_\alpha\left(f_1;f_2\right)=\frac{A\left(\eta_1\right)-A\left(\eta_h\right)+\ln E_h}{1-\alpha}-A\left(\eta_2\right), \label{eq:three}
\end{equation}
where $E_h= \mathbb{E}_{f_h}\left[b(X)^{\alpha-1}\right]=\int b(x)^{\alpha-1}f_h(x) \, dx$ 
\end{lemma}
\begin{proof}
Using \eqref{exp-f2}, we have
\begin{align*}
f_1&\left(x\right)f_2\left(x\right)^{\alpha-1}\\
&=b\left(x\right)\exp\big(\eta_1\cdot T\left(x\right)+A\left(\eta_1\right)\big)\\&\hspace{0.15 in}\cdot\bigg(b\left(x\right)\exp\big(\eta_2\cdot T\left(x\right)+A\left(\eta_2\right)\big)\bigg)^{\alpha-1}\\
&=b\left(x\right)^\alpha \exp\left((\eta_1+(\alpha-1)\eta_2)\cdot T\left(x\right)\right)\\&\hspace{0.15 in}\cdot\exp\left(A\left(\eta_1\right)+\left(\alpha-1\right)A\left(\eta_2\right)\right)\\
&=b\left(x\right)^\alpha \exp\left(\eta_h\cdot T\left(x\right)+A\left(\eta_h\right)\right)\\&\hspace{0.15 in}\cdot\exp\left(A\left(\eta_1\right)+\left(\alpha-1\right)A\left(\eta_2\right)-A\left(\eta_h\right)\right)\\
&=b\left(x\right)^{\alpha-1}f_h\left(x\right)\exp\left(A\left(\eta_1\right)+\left(\alpha-1\right)A\left(\eta_2\right)-A\left(\eta_h\right)\right).
\end{align*}
Thus,
\begin{align*}
  \int_{\mathbb{S}} &f_1\left(x\right)f_2\left(x\right)^{\alpha-1}dx\\&=\int_{\mathbb{S}}b\left(x\right)^{\alpha-1}f_h\left(x\right)dx\\&\hspace{0.15 in}\cdot\exp\left(A\left(\eta_1\right)+\left(\alpha-1\right)A\left(\eta_2\right)-A\left(\eta_h\right)\right)\\
    &=\exp\left(A\left(\eta_1\right)+\left(\alpha-1\right)A\left(\eta_2\right)-A\left(\eta_h\right)\right)E_h,
\end{align*}
and therefore,
$$h_\alpha\left(f_1;f_2\right) =\frac{A\left(\eta_1\right)-A\left(\eta_h\right)+\ln E_h}{1-\alpha}-A\left(\eta_2\right).$$
\end{proof}

\smallskip

\begin{remark} 
If $b(x)=b$ is a constant for all $x \in \mathbb{S}$, then 
$$\frac{\ln E_h}{1-\alpha}= -\ln b.$$ In many cases, we have that $b(x)= 1$ on $\mathbb{S}$, and thus the $\frac{\ln E_h}{1-\alpha}$ term disappears in~\eqref{eq:three}.
\end{remark}

\bigskip
Table \ref{RenDivTable} lists R\'{e}nyi differential cross-entropy expressions we derived using Lemma~\ref{lemma1} for some common distributions in the exponential family (which we describe in Appendix~B for convenience). In the table, the subscript of $i$ is used to denote that a parameter belongs to pdf $f_i$, $i=1,2$. 

\medskip

\begin{table}[hbtp]\caption{R\'{e}nyi Differential Cross-Entropies for Common Continuous Distributions}\label{RenDivTable} 
\begin{center}
\begin{tabular}{|c|c|} 

\hline
\textbf{Name} & $h_\alpha(f_1;f_2)$ \\ \hline
\hline

\textbf{Beta} & $\displaystyle{\ln {B(a_2, b_2)} +  \frac{1}{\alpha - 1}\ln \frac{B(a_h, b_h)}{B(a_1, b_1)}} $\\\cline{2-2} 
& $a_h := a_1 + (\alpha-1)(a_2-1)$,\\& $b_h := b_1 + (\alpha-1)(b_2-1)$ \\ 

\hline
$\boldsymbol{\chi^2}$ & $\displaystyle{\frac{1}{1-\alpha}\left( \frac{\nu_1}{2}\ln\left(\alpha\right)-\ln\Gamma\left(\frac{\nu_1}{2}\right)+\ln\Gamma\left(\frac{\nu_h}{2}\right)\right)}$\\
&$+\displaystyle{\frac{2-\nu_2}{2}\ln\left(\alpha \right)+\ln2\Gamma\left(\frac{\nu_2}{2}\right)} $  \\\cline{2-2} 
	& $\nu_h := \nu_1 + (\alpha-1)(k-2)$\\

\hline

\textbf{Exponential} & $\displaystyle{\frac{1}{1-\alpha}\ln \frac{\lambda_i}{\lambda_h} -\ln \lambda_2 }$  \\\cline{2-2} 
& $\lambda_h := \lambda_1 + (\alpha-1)\lambda_2$ \\

\hline 
\textbf{Gamma} & $\displaystyle{\ln \Gamma(k_2)+ k_2\ln\theta_2}$\\
& $+ \displaystyle{ \frac{1}{1-\alpha}\left( \ln\frac{\Gamma(k_h)}{\Gamma(k_1)}-k_h\ln\theta_h-k_1\ln\theta_1 \right) } $ \\\cline{2-2} 
&  $\theta_\alpha^* := \frac{\theta_1 + (a-1)\theta_2}{(\alpha-1)\theta_1\theta_1}$, \ $k_h := k_i + (\alpha-1)k_2$  \\ 
\hline 
\hspace{-8pt}
 \textbf{Gaussian}
 & $\displaystyle{\frac{1}{2}\left(\ln (2\pi\sigma_2^2) +  \frac{1}{1-\alpha}\ln \left( \frac{\sigma_2^2}{(\sigma^2)_h^*}\right) + \frac{(\mu_1 - \mu_2^2}{ (\sigma^2)_h^*} \right)} $\\\cline{2-2} 
& $(\sigma^2)_h^* := \sigma_2^2 + (\alpha-1)\sigma_1^2$ \\

\hline
$\begin{array}{c}
\textbf{Laplace}\\
\end{array}$
& $\displaystyle{\ln (2 b_2) +  \frac{1}{1-\alpha}\ln\left(\frac{b_2}{2b_h}\right)}$
\\\cline{2-2} 
\boldmath{($\mu_1=\mu_2$)}
 & $\displaystyle{b_h:=b_2+(1-\alpha)b_1}$\\
\hline
\end{tabular}

\end{center}
\end{table}

\section{R\'{e}nyi differential Cross-Entropy between different distributions}
Let $p$ and $q$ be pdfs with common support $\mathbb{S} \subseteq\mathbb{R}$. Below are some general formulae for the differential R\'{e}nyi cross-entropy between one specific (common) distribution and any general distribution. If $\mathbb{S}$ is an interval below, then $|\mathbb{S}|$ denotes its length.

\subsection{Distribution $q$ is uniform}

Let $q$ be uniformly distributed on $\mathbb{S}$. Then
\begin{align*}
    h_\alpha(p;q)&=\frac{1}{1-\alpha}\ln \int_\mathbb{S} p(x)q(x)^{\alpha-1}dx
    = \ln |\mathbb{S}|.
\end{align*}

\subsection{Distribution $p$ is uniform}

Now suppose $p$ is uniformly distributed on $\mathbb{S}$. Then
\begin{align*}
    h_\alpha(p;q)&=\frac{1}{1-\alpha}\ln \int_\mathbb{S} p(x)q(x)^{\alpha-1}dx\\
    &=\frac{1}{1-\alpha}\ln \frac{1}{|\mathbb{S}|} -h_{\alpha-1}(q).
\end{align*}

\subsection{Distribution $q$ is exponentially distributed}

Suppose the $\mathbb{S}=\mathbb{R}^+$ and $q$ is exponential with parameter $\lambda$. Suppose also that the moment generating function (MGF) of $p$, $M_p(t)$ exists. 
We have 
\begin{align*}
    h_\alpha(p;q)&=\frac{1}{1-\alpha}\ln \int_\mathbb{S} p(x)q(x)^{\alpha-1}dx\\
    &= \frac{1}{1-\alpha}\ln \mathbb{E}_p\left[q(x)^{\alpha-1}\right]\\
    &=\frac{1}{1-\alpha}\ln \mathbb{E}_p\left[\left(\lambda \exp\left(-\lambda x\right)\right)^{\alpha-1}\right]\\
    &=-\ln\lambda +\frac{1}{1-\alpha}\ln M_p\left(\lambda(1-\alpha)\right).
\end{align*}

\subsection{Distribution $q$ is Gaussian}
Now assume that $q$ is a (normal) Gaussian $\mathcal{N}(\mu,\sigma^2)$ distribution and that the MGF of $Y:= (X-\mu)^2$, $M_Y$, exists, where $X$ is a random variable with distribution $p$. Then 
\begin{align*}
    h_\alpha(p;q) &= \frac{1}{1-\alpha}\ln \mathbb{E}_p\left[q(X)^{\alpha-1}\right]\\
    &=\frac{1}{1-\alpha}\ln \sigma(\sqrt{2\pi})^{1-\alpha}\mathbb{E} \left(\exp\left((1-\alpha)\frac{Y}{2\sigma^2}\right)\right)\\
    &= \ln\sigma\sqrt{2\pi} +\frac{1}{1-\alpha}\ln M_Y\left(\frac{1-\alpha}{2\sigma^2}\right).
\end{align*}
The case where $q$ is a half-normal distribution can be directly derived from the above. Given $q$ is a half-normal distribution, on its support its pdf is the same as that of a normal $\mathcal{N}(0,\sigma^2)$ distribution times 2. Hence if $p$'s support is $\mathbb{R}^+$, then $h_\alpha(p;q)=\ln\sigma\sqrt{\frac{\pi}{2}} +\frac{1}{1-\alpha}\ln M_Y\left(\frac{1-\alpha}{2\sigma^2}\right)$.

\section{R\'{e}nyi Differential Cross-Entropy Rate for Stationary Gaussian Processes}

\begin{lemma}\label{lemma2}
The R\'{e}nyi differential cross-entropy between two zero-mean multivariate dimension-$n$ Gaussian distributions with invertible covariance matrices $\Sigma_1$ and $\Sigma_2$, respectively, is given by 
\begin{equation}\label{multi-gauss}
h_\alpha(p;q)=\frac{\ln |\Sigma_1||S|}{2\alpha-2}+ \frac{1}{2}\ln |\Sigma_2| +\frac{n}{2}\ln 2\pi,    
\end{equation}
where $S:=\Sigma_1^{-1}+(\alpha-1)\Sigma_2^{-1}$.
\end{lemma}
\begin{proof}
Recall that the pdf of a multivariate Gaussian with mean $\textbf{0}= (0,0..., 0)^T$ and invertible covariance matrix $\Sigma$ is given by:
$$f(\textbf{x})= \frac{\exp(\frac{-1}{2}\textbf{x}^T\Sigma^{-1}\textbf{x})}{(2\pi)^{k/2}|\Sigma|^{1/2}} $$
for $\textbf{x} \in \mathbb{R}^n$. Note that this distribution is a member of the exponential family, where $T(\textbf{x})=\textbf{x}$, $\eta = \frac{1}{2}\Sigma^{-1}$, $A(\eta)=\frac{1}{2}\ln |-2\eta|$ and $b(\textbf{x})=(2\pi)^{\frac{-n}{2}}$. Hence the R\'{e}nyi differential cross-entropy between two zero-mean multivariate Gaussian distributions with covariance matrices $\Sigma_1$ and $\Sigma_2$, respectively, is 
\begin{align*}
h_\alpha(p;q)&= 
    \frac{1}{1-\alpha}\left(\frac{1}{2}\ln \left|2\frac{\Sigma_1^{-1}}{2}\right| \right. \\
    & \qquad \left. -\frac{1}{2}\ln \left|2\frac{\Sigma_1^{-1}+(\alpha-1)\Sigma_2^{-1}}{2}\right| \right)\\ 
    & \qquad-\frac{1}{2}\ln \left|2\frac{\Sigma_2^{-1}}{2}\right| -\ln (2\pi)^{\frac{-n}{2}} \\
    &= \frac{\ln |\Sigma_1||S|}{2\alpha-2} + \frac{1}{2}\ln |\Sigma_2| +\frac{n}{2}\ln 2\pi.
\end{align*}
\end{proof}

\medskip

Let $\{X_j\}_{j=1}^\infty$ and $\{Y_j\}_{j=1}^\infty$ be stationary zero-mean Gaussian processes. For a given $n$, $X^n :=(X_1, X_2,...,X_n)$ and $Y^n :=(Y_1, Y_2,...,Y_n)$ are multivariate Gaussian random variables with mean \textbf{0} and covariance matrices $\Sigma_{X^n}$ and $\Sigma_{Y^n}$, respectively. Since $\{X_j\}$ and $\{Y_j\}$ are stationary, their covariance matrices are Toeplitz. Furthermore, $B^n:=\Sigma_{Y^n}+(\alpha-1)\Sigma_{X^n}$ is Toeplitz.
\begin{lemma}
Let $\tilde{f}(\lambda)$, $\tilde{g}(\lambda)$ and $\tilde{h}(\lambda)$ be the power spectral densities of $\{X_j\}$,
$\{Y_j\}$ and the zero-mean Gaussian process with covariance matrix $B^n$, respectively.

Then the R\'{e}nyi differential cross-entropy rate between $\{X_j\}$ and $\{Y_j\}$, $\lim_{n\to\infty} \frac{1}{n}h_\alpha(X^n;Y^n)$, is given by $$\frac{\ln 2\pi}{2} 
+\frac{1}{4\pi(1-\alpha)}\int_0^{2\pi} \left[(2-\alpha)\ln \tilde{g}(\lambda) -\ln \tilde{h}(\lambda)\right] \,d\lambda.
$$
\end{lemma}
\begin{proof}
From Lemma~\ref{lemma2}, we first note that $S = \Sigma_{X^n}^{-1}B_n\Sigma_{Y^n}^{-1}$. With this in mind the R\'{e}nyi differential cross-entropy can be rewritten using \eqref{multi-gauss} as
\begin{align*}
&\frac{1}{n}\bigg(\frac{\ln |\Sigma_{X^n}||\Sigma_{X^n}^{-1}B^n\Sigma_{Y^n}^{-1}|}{2(\alpha-1)} + \frac{1}{2}\ln |\Sigma_{Y^n}| +\frac{n}{2}\ln 2\pi\bigg)
\\
&=\frac{\ln2\pi}{2} +\frac{1}{2n}\bigg(\frac{\ln |\Sigma_{X^n}||\Sigma_{X^n}^{-1}||B^n||\Sigma_{Y^n}^{-1}|}{(\alpha-1)} + \ln |\Sigma_{Y^n}|\bigg)\\
&=\frac{\ln2\pi}{2} +\frac{1}{2n}\bigg(\frac{\ln |B^n|-\ln|\Sigma_{Y^n}|}{(\alpha-1)} + \ln |\Sigma_{Y_n}|\bigg)\\
    &= \frac{\ln 2\pi}{2}+ \frac{1}{2n(1-\alpha)}\left(\left(2-\alpha\right)\ln|\Sigma_{Y^n}|-\ln|B^n|\right).
\end{align*}
It was proven in \cite{Gray} that for a sequence of Toeplitz matrices $T_n$ with spectral density $t(\lambda)$ such that $\ln t(\lambda)$ is Reimann integrable, one has $$\lim_{n\to\infty}\ln |T^n| = \frac{1}{2\pi}\int_0^{2\pi}\ln t(\lambda)\,d\lambda.$$
We therefore obtain that the R\'{e}nyi differential cross-entropy rate is given by
\begin{align*}
    \frac{\ln 2\pi}{2}+ \frac{1}{4\pi(1-\alpha)}\int_0^{2\pi} \left[(2-\alpha)\ln \tilde{g}(\lambda) -\ln \tilde{h}(\lambda) \right] \,d\lambda.
\end{align*}
Note that $\tilde{h}(\lambda)= \tilde{g}(\lambda)+(\alpha-1)\tilde{f}(\lambda)$.
\end{proof}

\bigskip
 
\section{R\'{e}nyi cross-entropy rate for Markov sources}
Consider two time-invariant Markov sources $\{X_j\}_{j=1}^\infty$ and $\{Y_j\}_{j=1}^\infty$ with common finite alphabet $\mathbb{S}$ and with transition distribution $P(\cdot|\cdot)$ and $Q(\cdot|\cdot)$, respectively.
Then for any $i^n=(i_1,\ldots,i_n) \in \mathbb{S}^n$, their
$n$-dimensional joint distributions are given by
$$p^{(n)}(i^n)=P(i_n|i_{n-1})P(i_{n-1}|i_{n-2})...P(i_2|i_1)q(i_1)$$ and $$q^{(n)}(i^n)=Q(i_n|i_{n-1})Q(i_{n-1}|i_{n-2})...Q(i_2|i_1)p(i_1),$$
 respectively, with arbitrary initial distributions, $p(i_1)$ and $q(i_1)$, $i_1 \in \mathbb{S}$.
For simplicity, we assume that $p(i)$, $q(i),Q(j|i)>0$ for all $i,j\in \mathbb{S}$.
Define the R\'{e}nyi cross-entropy rate between $\{X_j\}$ and $\{Y_j\}$ as
\begin{align*}
  &\lim_{n\to \infty} \frac{1}{n}H_\alpha(X^n;Y^n) \\
  = &\lim_{n\to \infty} \frac{1}{n}\frac{1}{1-\alpha} \ln \left( \sum_{i^n\in \mathbb{S}^n}p^{(n)}(i^n)q^{(n)}(i^n)^{\alpha-1}\right).
\end{align*}
Note that by defining the matrix $R$ using the formula $$R_{ij}=P(j|i)Q(j|i)^{\alpha-1}$$ and the row vector $\textbf{s}$ as having components $s_i=p(i)q(i)^{\alpha-1}$, the R\'{e}nyi cross-entropy rate can be written as 
\begin{equation}\label{eq:markovtwo}
    \lim_{n\to \infty} \frac{1}{n}\frac{1}{1-\alpha}\ln \textbf{s} R^{n-1} \textbf{1},
\end{equation} where $\textbf{1}$ is a column vector whose dimension is the cardinatliy of the alphabet $\mathbb{S}$ and with all its entries equal to 1.

A result derived by \cite{rached2001renyi} for the R\'{e}nyi divergence between Markov sources can thus be used to find the R\'{e}nyi cross-entropy rate for Markov sources.
\begin{lemma}\label{lemma4}
Let $P$, $Q$, $\textbf{s}$ and $\textbf{R}$ be defined as above. If $R$ is irreducible, then 
\begin{equation}
\lim_{n\to \infty} \frac{1}{n}H_\alpha(X^n;Y^n)= \frac{\ln \lambda}{1-\alpha} \label{markov},
\end{equation}
where $\lambda$ is the largest positive eigenvalue of $R$.
\end{lemma}
\begin{proof}
Since the non-negative matrix $R$ is irreducible, by the Frobenius theorem (e.g., cf.\ \cite{seneta2006non, gallager1996}), it has a largest positive eigenvalue $\lambda$ with associated positive eigenvector $\textbf{b}$. Let $b_m$ and $b_M$ be the minimum and maximum elements, respectively, of $\textbf{b}$. Then due to the non-negativity of $\textbf{s}$, $$\lambda^{n-1}\textbf{s}\cdot \textbf{b}= \textbf{s}R^{n-1}\textbf{b}\leq \textbf{s}R^{n-1}\textbf{1}b_M,$$ where $\cdot$ denotes the Euclidean inner product. Similarly, $\lambda^{n-1}\textbf{s}\cdot \textbf{b} \geq \textbf{s}R^{n-1}\textbf{1}b_m.$ As a result,
$$\frac{1} {n} \ln \frac{\lambda^{n-1}\textbf{s}\cdot \textbf{b}}{b_M}\leq \frac{1} {n} \ln \textbf{s}R^{n-1}\textbf{1} \leq \frac{1} {n}\ln \frac{\lambda^{n-1}\textbf{s}\cdot \textbf{b}}{b_m}.$$ 
Note that for all $n$, $\frac{\textbf{s}\cdot \textbf{b}}{b_M}$ is a constant. Thus
\begin{align*}
    \lim_{n\to\infty} \frac{1} {n} \ln \frac{\lambda^{n-1}\textbf{s}\cdot \textbf{b}}{b_M}&=
   \lim_{n\to\infty} \frac{n-1}{n}\ln \lambda + \lim_{n\to\infty}\frac{1} {n} \ln \frac{\textbf{s}\cdot \textbf{b}}{b_M}\\
    &=\ln \lambda.
\end{align*}
Similarly, we have
$\lim_{n\to\infty}\frac{1} {n}\ln \frac{\lambda^{n-1}\textbf{s}\cdot \textbf{b}}{b_m}= \ln \lambda.$
Hence, 
$$\lim_{n\to \infty} \frac{1}{n}H_\alpha(X^n;Y^n)=
\lim_{n\to\infty}\frac{1} {n}\ln \frac{\lambda^{n-1}\textbf{s}\cdot \textbf{b}}{(1-\alpha)b_m}=\frac{\ln \lambda}{1-\alpha}.$$

\end{proof}

Another technique can be borrowed from \cite{rached2001renyi} to generalize Lemma~\ref{lemma4} to the case where $R$ is reducible. First $R$ is rewritten in the canonical form detailed in Proposition 1 of \cite{rached2001renyi}. Let $\lambda_k$ be the largest positive eigenvalue of each self-communicating sub-matrix of $R$, and let $\tilde{\lambda}$ be the maximum of these $\lambda_k$'s. For each inessential class $C_i$, let $\lambda_j$ be the largest positive eigenvalue of the sub-matrix of each class $C_j$ that is reachable from $C_i$, and let $\lambda^\dagger$ be the maximum of these $\lambda_j$'s. Define $\lambda = \max\{\tilde{\lambda},\lambda^\dagger\}$. Then~\eqref{markov} holds.


\section*{Appendix A: Shannon-type information measures}\label{app-a}

\begin{table}[h]
\begin{center}
\begin{tabular}{|c|c|} 

\hline
\textbf{Name} & Definition \\
\hline\hline

$\begin{array}{c}
\textbf{Shannon}\\
\textbf{Entropy}\\
\end{array}$ & $\displaystyle{H(p)=-\sum_{x\in\mathbb{S}}p(x)\ln p(x)} $\\
\hline
$\begin{array}{c}
\textbf{Shannon}\\ \textbf{Differential}\\
\textbf{Entropy}\\
\end{array}$& $\displaystyle{h(p)=-\int_{\mathbb{S}}p(x)\ln p(x) \, dx} $\\
\hline

$\begin{array}{c}
\textbf{Shannon}\\
\textbf{Cross-Entropy}\\
\end{array}$& $\displaystyle{H(p;q)=-\sum_{x\in\mathbb{S}}p(x)\ln q(x)} $\\
\hline
$\begin{array}{c}
\textbf{Shannon}\\ \textbf{Differential}\\
\textbf{Cross-Entropy}\\
\end{array}$
 & $\displaystyle{h(p;q)=-\int_{\mathbb{S}}p(x)\ln q(x) dx} $\\
\hline
$\begin{array}{c}
\textbf{KL Divergence,}\\
\textbf{(Discrete)}\\
\end{array}$ & $\displaystyle{D(p\|q)=-\sum_{x\in\mathbb{S}}p(x)\ln \frac{p(x)}{q(x)}} $\\
\hline
$\begin{array}{c}
\textbf{KL Divergence,}\\
\textbf{(Continuous)}\\
\end{array}$
& $\displaystyle{D(p\|q)=-\int_{\mathbb{S}}p(x)\ln \frac{p(x)}{q(x)} \, dx} $\\
\hline

\end{tabular}

\end{center}
\end{table}


\newpage
\section*{Appendix B: Distributions listed in Table \ref{RenDivTable}}\label{app-b}

\begin{table}[h]
\begin{center}
\begin{tabular}{|c|c|} 
\hline
\textbf{Name} & \textbf{PDF} $f(x)$ \\
\textit{(Parameters)} & (Support) \\\hline
\hline
\textbf{Beta} & $\displaystyle{{B(a, b)}x^{a-1}(1-x)^{b-1} } $\\
\textit{($a>0$, $b>0$)} & $\mathbb{S}=(0,1)$\\
\hline
$\boldsymbol{\chi^2}$ & $\displaystyle{ \dfrac{1}{2^{\frac \nu 2} \Gamma\left(\frac \nu 2 \right)}x^{\frac \nu 2 -1}e^{-\frac x 2}}$\\
\textit{($\nu \in \mathbb{Z}^+$)} & $\mathbb{S}=\mathbb{R^+}$\\
\hline 
\textbf{Exponential} & $\displaystyle{\lambda e^{-\lambda x}}$  \\
\textit{($\lambda>0$)} & $\mathbb{S}=\mathbb{R^+}$\\
\hline 
\textbf{Gamma} & $\displaystyle{\dfrac{1}{\theta^k\Gamma\left(k\right)}x^{k -1}e^{-\frac k \theta}}$\\
\textit{($k>0$, $\theta>0$)} & $\mathbb{S}=\mathbb{R^+}$\\
\hline 
 \textbf{Gaussian}
 & $\displaystyle{\frac{1}{ \sqrt{2\pi\sigma^2} } e^{-\frac{1}{2}\left(\frac{x-\mu}{\sigma}\right)^2}} $\\
\textit{($\mu$, $\sigma^2>0$)} & $\mathbb{S}=\mathbb{R}$\\
\hline
 \textbf{Laplace}
& $\displaystyle{\frac{1}{2b} e^{-\frac{|x-\mu|}{b}}}$
\\
\textit{($\mu$, $b^2>0$)} & $\mathbb{S}=\mathbb{R}$\\
\hline
\end{tabular}

\end{center}
\end{table}

\smallskip\noindent
{\bf Notes}

\begin{itemize}
    \item ${\displaystyle \mathrm {B}(a,b)=\int _{0}^{1}t^{a-1}(1-t)^{b-1}\,dt}$ is the Beta function.
    \item ${\displaystyle \Gamma(z)=\int _{0}^{\infty}x^{z-1}e^{-x}\,dx}$ is the Gamma function.
\end{itemize}

\bibliographystyle{IEEEtran}
\bibliography{citations}

\end{document}